\documentclass[12pt,a4paper]{article}
\usepackage{mathrsfs}
\usepackage{amsmath}
\usepackage{dsfont}
\usepackage{amsfonts}
\usepackage{amssymb}
\usepackage{graphicx}
\usepackage{txfonts}
\newtheorem{proposition}{Proposition}
\newenvironment{proof}{\noindent\textit{Proof~~}}
{\nolinebreak[4]\hfill$\square$\\\par}
\DeclareMathOperator*{\argmax}{\arg\!\max}
\title{Remarks on quantum duopoly schemes }
\author{Piotr Fr\c{a}ckiewicz\\
\small Institute of Mathematics, Pomeranian University\\ \small 76-200 S\l upsk, Poland\\ \small P.Frackiewicz@impan.gov.pl}
\begin{document}
\maketitle
\begin{abstract}
The aim of this paper is to discuss in some detail the two different quantum schemes for duopoly problems. We investigate under what conditions one of the schemes is more reasonable that the other one. Using the Cournot's duopoly example we show that the current quantum schemes require a slight refinement so that they output the classical game in a particular case. Then we show how the amendment changes the way of studying the quantum games with respect to Nash equilibria. Finally, we define another scheme for the Cournot's duopoly in terms of quantum computation.
\end{abstract}
\section{Introduction}
\label{intro}
Quantum game theory, an interdisciplinary field that combines quantum theory and game theory, has been investigated for fifteen years. The first attempt to describe a game in the quantum domain applied to finite noncooperative games in the normal form \cite{mejer}, \cite{ewl}, \cite{marinatto}. The general idea (in the case of bimatrix games) was based on identifying the possible results of the game with the basis states $|ij\rangle \in \mathbb{C}^{n}\otimes \mathbb{C}^{m}$. Soon after quantum theory has also found an application in duopoly problems \cite{iqbalbackwards}, \cite{du}. It has been a challenging task as the players' strategy sets in the duopoly examples are (real) intervals and therefore there are continuum of possible game results. The scheme presented in \cite{iqbalbackwards} adapts the Marinatto-Weber quantum $2\times 2$ game scheme \cite{marinatto} to the Stackelberg duopoly example. The model relates actions in the duopoly with the probabilities of applying the bit-flip operators. On the other hand, the model introduced in \cite{du} is a new framework compared with the quantum $2\times 2$ game schemes. It was defined to consider the Cournot duopoly problem where each strategy in the classical game corresponds to a specific unitary operator. The model entangles the players' quantities and the relation depends on a degree of the entanglement. 

The Iqbal-Toor \cite{iqbalbackwards} and Li-Du-Massar \cite{du} schemes undoubtedly brought new ideas to the field of quantum games. These remarkable schemes have found an application in many duopoly problems. The former scheme was further investigated, for example, in \cite{zhu}, \cite{zhu2}, \cite{bertrand}, \cite{bertrand2}, \cite{doublekhan}, the latter one in \cite{1.1}, \cite{2.1} \cite{3.1}, \cite{3.3}, \cite{3.4}, \cite{4.1}, \cite{5}, \cite{6}, \cite{7}, \cite{8.1}, \cite{8.2}, \cite{8.3}.
The aim of the paper is to pay attention to some properties of the quantum duopoly schemes that might be thought unsuitable (in the case of the Iqbal-Toor scheme) and specify the Nash equilibrium analysis (in the case of the Li-Du-Massar scheme). 
\section{Cournot's duopoly model}
We recall the version of the Cournot model \cite{cournot} with two firms that have been the subject of research in the quantum domain. Based on \cite{peters}, firm 1 and firm 2 offer quantities $q_{1}$ and $q_{2}$, respectively, of a homogeneous product. The price of the product depends on the total quantity $q_{1} + q_{2}$. The higher the quantity is, the lower the price of the product. More formally, the Cournot's duopoly model defines a strategic form game $(N, \{S_{i}\}_{i\in N}, \{u_{i}\}_{i\in N})$, where 
\begin{enumerate}
\item $N=\{1,2\}$ is a set of players, \item $S_{i} = [0,\infty)$ is a player $i$'s strategy set, \item $u_{i}$ is a player $i$'s payoff function given by formula 
\begin{equation}\label{classicpayoff}
u_{i}(q_{1}, q_{2}) = q_{i}P(q_{1}, q_{2}) - cq_{i}  ~~\mbox{for}~~ q_{1}, q_{2} \in [0,\infty).
\end{equation}
Here $P(q_{1}, q_{2})$ represents the market price of the product, 
\begin{equation}\label{classicprice}
P(q_{1}, q_{2}) = \begin{cases}  a-q_{1} - q_{2} &\mbox{if}~ q_{1} + q_{2} \leqslant a \\ 0 &\mbox{if}~ q_{1} + q_{2} > a, \end{cases}
\end{equation}
and $c$ is a marginal cost with $a>c > 0.$
\end{enumerate}

A Nash equilibrium is the most commonly used solution concept to study duopoly examples. It is defined as a profile of strategies of all players in which each strategy is a best response to the other strategies. In view of the Cournot duopoly it is a strategy profile $(q^*_{1}, q^*_{2})$ that satisfies the following system of inequalities:
\begin{equation}
\begin{cases} u_{1}(q^*_{1}, q^*_{2}) \geqslant u_{1}(q_{1}, q^*_{2}) \\  u_{2}(q^*_{1}, q^*_{2}) \geqslant u_{2}(q^*_{1}, q_{2})\end{cases} ~\mbox{for all}~q_{1}, q_{2} \in [0,\infty).
\end{equation}
It implies that the game has a unique Nash equilibrium $(q^*_{1}, q^*_{2}) = ((a-c)/3, (a-c)/3)$ with the equilibrium payoff $(a-c)^2/9$ for each player.

It is appropriate at this point to note that in literature one can find (\ref{classicprice}) with requirement $a>c\geqslant 0$ and a statement that the Cournot duopoly has the unique Nash equilibrium. In fact, $c>0$ is crucial to the uniqueness of the equilibrium. If $c=0$ and one of the players, say player 1, chooses $q_{1}\geqslant a$ then the player 2's set of best replies is $[0,\infty)$. By completely symmetric arguments, $q_{1} \in [0,\infty)$ is player 1's best reply to $q_{2}\geqslant a$. Thus there would be continuum many equilibria $(q^*_{1}, q^*_{2})$ such that $q^*_{1}, q^*_{2} \geqslant a$ with payoff 0 for both players. 
\section{Remarks on existing quantum duopoly schemes}\label{section3}
In this section we discuss the two main quantum approaches to the problem of duopoly. Both schemes show how to define game with uncountable sets of strategies.
\subsection{The Iqbal-Toor quantum duopoly scheme}\label{sectionMW}
In paper~\cite{iqbalbackwards} the authors adapted the Marinatto-Weber quantum scheme for $2\times 2$ games to the problem of the Stackelberg's duopoly. We restrict ourselves to the Cournot's duopoly, where the players choose their actions simultaneously instead of a sequential order. It does not affect the framework of the quantum scheme but simplifies the analysis. The key idea relies on identifying players' actions $q_{1}, q_{2} \in [0,\infty)$ in the classical duopoly with probabilities that determine the final state in the quantum model of $2\times 2$ games. Then, by appropriately defined measurement operators, the scheme, in particular case, is supposed to output the classical game. Formally, the final state $\rho_{\mathrm{fin}}$  associated with the Iqbal-Toor scheme has the form
\begin{align}\label{finalstate}
\rho_{\mathrm{fin}} &= xy\mathds{1}\otimes \mathds{1} \rho_{\mathrm{in}} \mathds{1}\otimes \mathds{1} + x(1-y)\mathds{1}\otimes \sigma_{x}\rho_{\mathrm{in}}\mathds{1}\otimes \sigma_{x} \nonumber\\ & \quad + (1-x)y\sigma_{x}\otimes \mathds{1}\rho_{\mathrm{in}}\sigma_{x}\otimes \mathds{1} + (1-x)(1-y)\sigma_{x} \otimes \sigma_{x} \rho_{\mathrm{in}}\sigma_{x} \otimes \sigma_{x},
\end{align}
where $x$ and $(1-x)$ $(y ~\mbox{and}~ (1-y))$ are the probabilities of choosing by player 1 (player 2) the identity operator $\mathds{1}$ and the Pauli operator $\sigma_{x}$, respectively. In order to associate player $i$'s actions $q_{i} \in [0,\infty)$ for $i=1,2$ with final state (\ref{finalstate}) the authors defined the following probability relations:
\begin{equation}\label{probabilities}
x = \frac{1}{1+q_{1}},\quad y = \frac{1}{1+q_{2}}.
\end{equation}
As a result, if $\rho_{\mathrm{in}} = |00\rangle \langle 00|$ and the probabilities $x$ and $y$ are given by equation~(\ref{probabilities}), the final state~(\ref{finalstate}) can be written as
\begin{equation}\label{finmw}
\rho_{\mathrm{fin}} = \frac{1}{(1+ q_{1})(1+q_{2})}[|00\rangle \langle 00| + q_{2}|01\rangle \langle 01| + q_{1}|10\rangle \langle 10| + q_{1}q_{2}|11\rangle \langle 11|].
\end{equation}
Given the payoff operator $M_{i}$,
\begin{equation}\label{operatorm}
M_{i} = (1+q_{1})(1+q_{2})q_{i}[(a-c)|00\rangle \langle 00|-|01\rangle \langle 01| - |10\rangle \langle 10|],
\end{equation}
player $i$'s payoff is of the form
\begin{equation}\label{payoff}
u_{i}(q_{1}, q_{2}) = \mathrm{tr}(\rho_{\mathrm{fin}}M_{i}) = q_{i}(a-c - q_{1} - q_{2})~~\mbox{for}~~i=1,2.
\end{equation} 

Let us consider scheme~(\ref{finalstate})-(\ref{payoff}) in view of the Cournot's duopoly. The first problem we can see is that model~(\ref{finalstate})-(\ref{payoff}) in fact does not reproduce the game defined by (\ref{classicpayoff})-(\ref{classicprice}). The payoff function~(\ref{payoff}) coincides with formula (\ref{classicpayoff}) for $q_{1} + q_{2} \leqslant a$ but it does not take into account the market price $P(q_{1},q_{2})$ equal to zero in the case $q_{1} + q_{2} > a$. As a result, if, for example, $q_{2} > a$, player 1's payoff function in the classical case comes down to $-cq_{1}$ and it is different in general from $q_{1}(-a-c - q_{1})$ given by (\ref{payoff}). 

The scheme (\ref{finalstate})-(\ref{payoff}) will generalize the classical Cournot's duopoly if we modify operator~(\ref{operatorm}) to include the case $q_{1} + q_{2} > a$. Let us define
\begin{equation}\label{poprawionym}
M'_{i} = \begin{cases} (1+ q_{1})(1+q_{2})q_{i}[(a-c)|00\rangle \langle 00| - |01\rangle \langle 01| - |10\rangle \langle 10|] &if~~q_{1} + q_{2} \leqslant a \\ (1+ q_{1})(1+q_{2})q_{i}(-c|00\rangle \langle 00|) &if~~q_{1} + q_{2} > a. \end{cases}
\end{equation} 
Then $\mathrm{tr}(\rho_{\mathrm{fin}}M'_{i})$, where $\rho_{\mathrm{fin}}$ is given by~(\ref{finmw}), determines the complete payoff function~(\ref{classicpayoff}). 

It is worth noting that there are many ways to define the measurement operator $M_{i}$ that determines the same outcome for $\rho_{\mathrm{in}} = |00\rangle \langle 00|$. Let operators~(\ref{operatorm}) for $i=1,2$ be replaced by
\begin{equation}\label{m2prim}
\begin{array}{l}
M''_{1} = (1+q_{1})(1+ q_{2})[(a-c)q_{1}|00\rangle \langle 00| - q_{1}|10\rangle \langle 10| - |11\rangle \langle 11|] \\ M''_{2} = (1+q_{1})(1 + q_{2})[(a-c)q_{2}|00\rangle \langle 00| - q_{2}|01\rangle \langle 01| - |11\rangle \langle 11|].
\end{array}
\end{equation}
Then $\mathrm{tr}(\rho_{\mathrm{fin}}M''_{i}) = \mathrm{tr}(\rho_{\mathrm{fin}}M_{i})$ for $\rho_{\mathrm{in}} = |00\rangle \langle 00|$. Simultaneously, the payoff function defined by~(\ref{m2prim}) is different from one determined by~(\ref{operatorm}) for other initial states. For example, in the case of the initial state $\rho_{\mathrm{in}} = |11\rangle \langle 11|$, the final state $\rho_{\mathrm{fin}}$ has the form
\begin{equation}\label{fin11}
\rho_{\mathrm{fin}} = \frac{1}{(1+q_{1})(1+q_{2})}\left[|11\rangle \langle 11| + q_{2}|10\rangle \langle 10| + q_{1}|01\rangle \langle 01|+ q_{1}q_{2}|00\rangle \langle 00|\right].
\end{equation}
Then, the payoff functions corresponding to measurements $M_{1}$ and $M''_{1}$ are given by formulae
\begin{equation}\label{roznefunkcje}
\mathrm{tr}(\rho_{\mathrm{fin}}M_{1}) = q_{1}\left[(a-c)q_{1}q_{2} - q_{1} - q_{2}\right], \quad \mathrm{tr}(\rho_{\mathrm{fin}}M''_{1}) = q_{1}\left[(a-c)q_{1}q_{2} - q_{2}\right] - 1.
\end{equation}

Equations~(\ref{roznefunkcje}) suggest that we can modify model (\ref{finalstate})-(\ref{payoff}) in various ways, each time obtain new results, for example, with respect to Nash equilibria. However, the scheme (\ref{finalstate})-(\ref{payoff}) has a feature that may question the correctness of the model. Namely, it outputs non-classical results already for separable initial states. In order to see that, let us use results~(\ref{fin11}) and (\ref{roznefunkcje}) obtained by letting $\rho_{\mathrm{in}} = |11\rangle \langle 11|$. Then for arbitrary $a-c > 0$ and suitably large $q_{2}>0$
expression $\mathrm{tr}(\rho_{\mathrm{fin}}M_{1})$ as a function of variable $q_{1}$ increases monotonically and $\lim_{q_{1} \to \infty}\mathrm{tr}(\rho_{\mathrm{fin}}M_{1}) = \infty$. In the same manner we can see that if $q_{1}$ is suitably large, the higher quantity $q_{2}$ is, the more player 2 gains. Therefore, scheme (\ref{finalstate})-(\ref{payoff}) with the initial state $\rho_{\mathrm{in}} = |11\rangle \langle 11|$ allows the players to obtain arbitrary large payoffs. By contrast, each player's set of available payoffs in the classical Cournot duopoly is bounded from above. The upper bound is $(a-c)^2/4$. It corresponds to the monopoly quantity $(a-c)/2$ given by solving $\argmax_{q_{1} \geqslant 0}u_{1}(q_{1},0)$.

An easy computation shows that there is no another form of $M_{i}$ to obtain the scheme (\ref{finalstate})-(\ref{payoff}) that outputs the classical game for $\rho_{\mathrm{in}} = |j_{1},j_{2}\rangle \langle j_{1},j_{2}|$, $j_{1},j_{2} = 0,1$ and is a nontrivial generalization of (\ref{classicpayoff}). Indeed, let us consider the general form of a payoff operator for $i=1,2$,
\begin{equation}\label{ogolnym}
M_{i} = (1+q_{1})(1+ q_{2})\left(x^{i}_{1}|00\rangle \langle 00| + x^{i}_{2}|01\rangle \langle 01| + x^{i}_{3}|10\rangle \langle 10| + x^{i}_{4}|11\rangle \langle 11|\right),
\end{equation}
where $x^{i}_{1}, x^{i}_{2}, x^{i}_{3}, x^{i}_{4} \in \mathbb{R}$. Suppose operator $M_{i}$, $i=1,2$ satisfies the equation
\begin{equation}\label{condition}
\mathrm{tr}\left(\rho_{\mathrm{fin}}M_{i}\right) = q_{i}(a-c -q_{1} - q_{2})~~ \mbox{for each}~~j_{1}, j_{2} =0,1.
\end{equation}
If $\rho_{\mathrm{in}} = |j_{1},j_{2}\rangle \langle j_{1},j_{2}|$ for $j_{1}, j_{2} = 0,1$, the final state $\rho_{\mathrm{fin}}$ described by (\ref{finalstate}) assumes the form
\begin{align}
\rho_{\mathrm{fin}} &= \frac{1}{(1+q_{1})(1+q_{2})}(|j_{1},j_{2}\rangle \langle j_{1},j_{2}| + q_{2}|j_{1}, j_{2}\oplus_{2}1\rangle \langle j_{1}, j_{2}\oplus_{2}1| \nonumber\\  &\quad+ q_{1}|j_{1}\oplus_{2} 1, j_{2}\rangle \langle j_{1}\oplus_{2} 1, j_{2}| + q_{1}q_{2}|j_{1}\oplus_{2}1, j_{2}\oplus_{2}1\rangle \langle j_{1}\oplus_{2}1, j_{2}\oplus_{2}1|),
\end{align}
where $\oplus_{2}$ means addition modulo 2. Then condition~(\ref{condition}) determines a system of four equations
\begin{equation}
\begin{cases}x_{1} + q_{2}x_{2} + q_{1}x_{3} + q_{1}q_{2}x_{4} - q_{i}(a-c-q_{1}-q_{2}) = 0 \\ q_{1}q_{2}x_{1} + q_{1}x_{2} + q_{2}x_{3} + x_{4} -q_{i}(a-c-q_{1}-q_{2}) = 0 \\ q_{2}x_{1} + x_{2} + q_{1}q_{2}x_{3} + q_{1}x_{4} - q_{i}(a-c-q_{1}-q_{2}) = 0 \\ q_{1}x_{1} + q_{1}q_{2}x_{2} + x_{3} + q_{2}x_{4} - q_{i}(a-c-q_{1}-q_{2}) = 0 \end{cases}
\end{equation}
that has the unique solution
\begin{equation}
x_{1}=x_{2} = x_{3} = x_{4} = \frac{q_{i}(a-c-q_{1}-q_{2})}{(1+q_{1})(1+q_{2})}.
\end{equation}
As a result, operator (\ref{ogolnym}) comes down to the trivial payoff operator
\begin{equation}
M_{i} = q_{i}(a-c - q_{1}-q_{2})\mathds{1}\otimes \mathds{1}
\end{equation}
that implies the outcome $\mathrm{tr}(\rho_{\mathrm{fin}}M_{i}) = q_{i}(a-c-q_{1} -q_{2})$ for each density operator $\rho_{\mathrm{in}}$ defined on $\mathbb{C}^{2}\otimes \mathbb{C}^{2}$.

Clearly, in the game determined by (\ref{finalstate})-(\ref{payoff}) and $\rho_{\mathrm{in}} = |11\rangle \langle 11|$ there is no profitable Nash equilibrium. For the fixed profile $(q_{1}, q_{2})$, player $i$ would benefit by choosing $q'_{i} > q_{i}$. Thus, one may question, if the players are able to obtain superior payoffs compared to the classical case. However, as it has been mentioned, that model does not take into account the proper price function~(\ref{classicprice}). It cannot then be considered in terms of a generalization of the Cournot model. Let us now replace~(\ref{operatorm}) with~(\ref{poprawionym}). This gives
\begin{equation}\label{function11}
u^{Q}_{i}(q_{1},q_{2}) = \mathrm{tr}(\rho_{\mathrm{fin}}M'_{i}) = \begin{cases}q_{i}\left[(a-c)q_{1}q_{2} - q_{1} - q_{2}\right] &\mbox{if}~q_{1} + q_{2} \leqslant a, \\ -cq_{1}q_{2} &\mbox{if}~q_{1} + q_{2} > a. \end{cases}
\end{equation}
Let us assume now that $a \geqslant (c + \sqrt{c^2 + 16})/2$ and consider a strategy profile $(q^*_{1}, q^*_{2}) = (a/2, a/2)$. If $q_{1} \leqslant a/2$ then
\begin{align}
u^{Q}_{1}\left(\frac{a}{2}, \frac{a}{2}\right) - u^{Q}_{1}\left(q_{1}, \frac{a}{2}\right) &= \left(\frac{a}{2} - q_{1}\right)\left[\left(\frac{a}{2}(a-c)-1\right)\left(\frac{a}{2}+q_{1}\right) - \frac{a}{2}\right] \nonumber\\ &\geqslant \left(\frac{a}{2} - q_{1}\right)\left(\frac{a}{2} + q_{1} - \frac{a}{2}\right) =\left(\frac{a}{2} - q_{1}\right)q_{1} \geqslant 0,
\end{align}
where the second to last inequality follows from the fact that the assumption $a \geqslant (c + \sqrt{c^2 + 16})/2$ implies $(a/2)(a-c)-1 \geqslant 1.$ In the case $q_{1} > a/2$,
\begin{align}
u^{Q}_{1}\left(q_{1},\frac{a}{2}\right) = -cq_{1}\frac{a}{2} &<0= \left(\frac{a}{2}\right)^2 - \left(\frac{a}{2}\right)^2 \nonumber \\ &\leqslant \left(\frac{a}{2}\right)^2\left(\frac{a}{2}(a-c)-1\right) - \left(\frac{a}{2}\right)^2 = u^{Q}_{1}\left(\frac{a}{2}, \frac{a}{2}\right).
\end{align}
We have thus proved that $q^*_{1} = a/2$ is a player 1's best response to $q^*_{2} = a/2$. In an exactly similar way we can show that $q^*_{2} = a/2$ is a best response to $q^*_{1} = a/2$. This means that $(q^*_{1}, q^*_{2})$ is a Nash equilibrium. We now conclude from comparing the equilibrium payoffs in the classical Cournot duopoly and $u^{Q}_{i}(a/2,a/2)$ defined by function~(\ref{function11}) that the latter payoff can be much greater. 
\begin{figure}[t]
\centering
\includegraphics[scale=1]{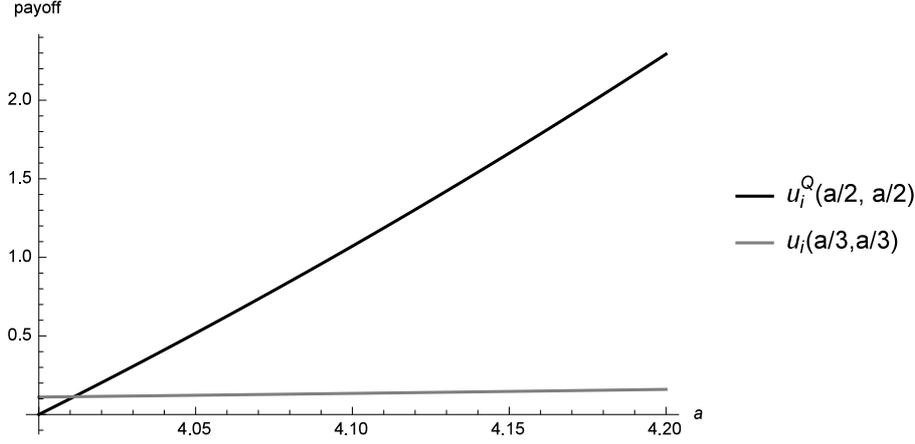}
\caption{The graphs of $u_{i}(a/3,a/3)$ and $u^{Q}_{i}(a/2,a/2)$ for initial numbers $a \geqslant (c + \sqrt{c^2 + 16})/2$, where $c = 3$. \label{figure0}}
\end{figure}
For example (see also the graphs in Fig~\ref{figure0}), if $a=30$ and $c=3$, the classical equilibrium outcome is 81 whereas $u^{Q}_{i}(a/2,a/2) = 90675$.

It is worth noting that the problem of nonclassical results determined by the basis states $\rho_{\mathrm{in}}$ spreads to other duopoly examples based on scheme (\ref{finalstate})-(\ref{payoff}). Let us recall the problem of Bertrand duopoly studied in~\cite{bertrand}. In this case two firms (players) compete in price. Let $p_{1}$ and $p_{2}$ be the prices chosen by player 1 and 2, respectively. Then the product quantities are given by
\begin{equation}\label{b1}
q_{1} = a - p_{1} + bp_{2},~~ q_{2} = a - p_{2} + bp_{1},~~\mbox{where}~~0<b<1,
\end{equation}
and the payoff functions are
\begin{equation}\label{b2}
u_{1}(p_{1}, p_{2}) = (a-p_{1} + bp_{2})(p_{1}-c),~~ u_{2}(p_{1}, p_{2}) = (a-p_{2} + bp_{1})(p_{2}-c).
\end{equation}
The quantum scheme for the problem~(\ref{b1})-(\ref{b2}) proceeds similarly to (\ref{finalstate})-(\ref{payoff}) where each $q_{i}$ in (\ref{probabilities}) is replaced with $p_{i}$ (for the detailed description we refer the reader to \cite{bertrand}). It was designed to consider the initial state of the form
\begin{eqnarray}
|\Psi_{\mathrm{in}}\rangle = \cos{\gamma}|00\rangle + \sin{\gamma}|11\rangle,~~\gamma \in [0,2\pi). \nonumber
\end{eqnarray}
Then the resulting payoff functions $u^{Q}_{1}$ and $u^{Q}_{2}$ depend on a triple $(p_{1}, p_{2}, \gamma)$ and are equal to
\begin{eqnarray}\label{bertrandpayoff}
&&u^{Q}_{1}(p_{1}, p_{2}) = (a-p_{1} + bp_{2})\left[(p_{1} - c)\cos^2{\gamma} + \left(p_{2} + p_{1}\left(-1 -cp_{2} + p^2_{2}\right)\right)\sin^{2}{\gamma}\right], \nonumber \\ &&u^{Q}_{2}(p_{1}, p_{2}) = (a-p_{2} + bp_{1})\left[(p_{2} - c)\cos^2{\gamma} + \left(p_{1} + p_{2}\left(-1 -cp_{1} + p^2_{1}\right)\right)\sin^{2}{\gamma}\right]. \nonumber
\end{eqnarray}
Thus, in particular, if $\rho_{\mathrm{in}} = |11\rangle \langle 11|$ and we consider strategy profile in the form of $(0,p_{2})$, $p_{2}>0$ the formulae $u^{Q}_{1}(p_{1}, p_{2})$ comes down to 
\begin{equation}
u^{Q}_{1}(p_{1}, p_{2}) = (a + bp_{2})p_{2},~~u^{Q}_{2}(p_{1}, p_{2}) = -(a-p_{2})p_{2}.
\end{equation}
Now, it is easily seen that the players' payoffs approach infinity as $p_{2}$ approaches infinity. This is clearly a nonclassical result as in the classical case~(\ref{b2})
\begin{equation}
\sup_{p_{1}, p_{2} \geqslant 0}\left(u_{1}(p_{1},p_{2}) + u_{2}(p_{1},p_{2})\right) = \frac{[a-c(1-b)]^2}{2(1-b)} < \infty.
\end{equation}

The question now arises: given the results above, can the quantum scheme (\ref{finalstate})-(\ref{poprawionym}) be considered reasonable? The requirement taken from quantum schemes for finite strategic games is satisfied: the duopoly example defined by~(\ref{classicpayoff})-(\ref{classicprice}) can be obtained from (\ref{finalstate})-(\ref{poprawionym}). Thus, in this sense, the quantum duopoly based on the Marinatto-Weber scheme appears to be well-defined. However, there is also some implicit condition embeded in the framework for quantum playing finite strategic games. Namely, each payoff profile generated by the quantum scheme lies in the convex hull of set of payoff profiles
determined by the classical game. For example, in the Marinatto-Weber and the Eisert-Wilkens-Lewenstein schemes for $2\times 2$ games the payoff operator has the form $M = \sum_{j_{1},j_{2}=0,1} x_{j_{1},j_{2}}|j_{1},j_{2}\rangle \langle j_{1},j_{2}|$, where $x_{j_{1},j_{2}} \in \mathbb{R}^2$ are the four payoff pairs that define the $2\times 2$ game. Then for any density operator $\rho_{\mathrm{fin}}$ on $\mathbb{C}^2\otimes \mathbb{C}^2$ determined by the players, the resulting payoff pair $\mathrm{tr}(\rho_{\mathrm{fin}}M)$ is a convex combination of points $x_{j_{i},j_{2}}$. If we assumed that equal convex hulls of payoff profiles generated by both the classical and quantum game were a necessary condition for the quantum scheme to be a correct one, then the protocol (\ref{finalstate})-(\ref{poprawionym}) would not be valid.
\subsection{The Li-Du-Massar quantum duopoly scheme}
The quantum protocol introduced in \cite{du} is another way to define the problem of duopoly in the quantum domain. Let us recall the formal description of this scheme \cite{du,hotel}. 

Let $|\Psi_{\mathrm{in}}\rangle = |00\rangle$ be the initial state and $J(\gamma) = \exp\left\{-\gamma\left(a^{\dag}_{1}a^{\dag}_{2}-a_{1}a_{2}\right)\right\}$, $\gamma \geqslant 0$, where $a^{\dag}_{i}$ $(a_{i})$ is the creation (annihilation) operator of player i's electromagnetic field. The player $i$'s strategies depend on $x_{i}\in [0,\infty)$ and they are given by formula
\begin{equation}\label{operatorydu}
D_{i}(x_{i}) = \exp\left\{\frac{x_{i}(a^{\dag}_{i}-a_{i})}{\sqrt{2}}, x_{i}\in [0,\infty)\right\},~i=1,2.
\end{equation} 
A quantum measurement on state $|\Psi_{\mathrm{fin}}\rangle$, 
\begin{equation}
|\Psi_{\mathrm{fin}}\rangle = J(\gamma)^{\dag}(D_{1}(x_{1})\otimes D_{2}(x_{2}))J(\gamma)|00\rangle,
\end{equation}
described by the observables $X_{i}=\left(a^{\dag}_{i} + a_{i}\right)/\sqrt{2}$, $i=1,2$ gives the quantities
\begin{equation}\label{quantumquantities}
\begin{array}{l}
q_{1} = \langle \Psi_{\mathrm{fin}}|X_{1}|\Psi_{\mathrm{fin}}\rangle = x_{1}\cosh{\gamma} + x_{2}\sinh{\gamma},\cr q_{2} = \langle \Psi_{\mathrm{fin}}|X_{2}|\Psi_{\mathrm{fin}}\rangle = x_{2}\cosh{\gamma} + x_{1}\sinh{\gamma}.
\end{array}
\end{equation}
Equations~(\ref{quantumquantities}) are obtained by using the formula (see also \cite{quantumbook})
\begin{eqnarray}
\mathrm{e}^{\lambda A}B\mathrm{e}^{-\lambda A} = B\cosh{\lambda\sqrt{\beta}} + \frac{[A,B]}{\sqrt{B}}\sinh{\lambda\sqrt{\beta}}, \nonumber
\end{eqnarray}
where the operators $A$ and $B$ satisfy the relation $[A, [A,B]] = \beta B$,~ $\beta$-constant. Given~(\ref{quantumquantities}) the payoff functions are
\begin{equation}\label{payoffli}
\begin{array}{l}
u_{1}^{Q}(x_{1},x_{2}) = (x_{1}\cosh{\gamma} + x_{2}\sinh{\gamma})[a-c-\mathrm{e}^{\gamma}(x_{1} + x_{2})] \cr u_{2}^{Q}(x_{1},x_{2}) = (x_{2}\cosh{\gamma} + x_{1}\sinh{\gamma})[a-c-\mathrm{e}^{\gamma}(x_{1} + x_{2})]
\end{array}
\end{equation}
that simply follow from the payoff function (\ref{classicpayoff}) for $q_{1} + q_{2} \leqslant a$. 

Note that, in contrast to the previous scheme (\ref{finalstate})-(\ref{payoff}), the Li-Du-Massar scheme does not act on the payoff function $u_{i}$ but merely defines a new relation between the players' choices $x_{1}, x_{2}$ and the resulting quantities $q_{1}, q_{2}$. The values $q_{1}, q_{2}$ given by (\ref{quantumquantities}) are still the real numbers from set $[0,\infty)$. It implies that the set of available payoffs determined by the Li-Du-Massar scheme coincides with the one of the classical problem. However, similarly to (\ref{finalstate})-(\ref{payoff}), the Li-Du-Massar scheme does not take into account the complete market price function (\ref{classicpayoff}), i.e., the case $P(q_{1}, q_{2}) = 0$ if $q_{1} + q_{2}>a$. As a consequence, one cannot say that model defined by (\ref{operatorydu})-(\ref{payoffli}) generalizes the Cournot duopoly problem. Reffering to (\ref{classicpayoff}), the payoff function $u_{1(2)}^{Q}(x_{1},x_{2})$ should be extended to take the form 
\begin{equation}\label{poprawionypayoff}
u_{1(2)}^{Q}(x_{1},x_{2}) = \begin{cases}\left(x_{1(2)}\cosh{\gamma} + x_{2(1)}\sinh{\gamma}\right)\left[a-c - \mathrm{e}^{\gamma}(x_{1}+x_{2})\right] &\mbox{if}~ \mathrm{e}^{\gamma}(x_{1}+x_{2}) \leqslant a, \cr -c\left(x_{1(2)}\cosh{\gamma} + x_{2(1)}\sinh{\gamma}\right) & \mbox{if}~ \mathrm{e}^{\gamma}(x_{1}+x_{2}) > a.\end{cases}
\end{equation}
Then if $\gamma = 0$, formula (\ref{poprawionypayoff}) comes down to (\ref{classicpayoff}). 

The question is now: whether the quantum scheme with payoff function (\ref{payoffli}) and (\ref{poprawionypayoff}) have different sets of Nash equilibria. Similarly to the classical case, if $c=0$, each profile $(x^*_{1},x^*_{2})$ such that $x_{1}, x_{2} \geqslant \mathrm{e}^{-\gamma}a$ is a Nash equilibrium in the case (\ref{poprawionypayoff}). The players obtain the payoff 0 and any unilateral deviation from the equilibrium strategy does not change the player's payoff.  If $c>0$ then there is a unique Nash equilibrium $(x^*_{1}, x^*_{2})$ that coincides with the one determined in \cite{du}.
However, the proof of existence and uniqueness of the equilibrium needs a more sophisticated reasoning compared with \cite{du}. In what follows, we give a rigorous proof of the following fact:
\begin{proposition}\label{prop1}
If the marginal cost $c$ is positive, the quantum Cournot duopoly defined by the Li-Du-Massar scheme with the payoff function (\ref{poprawionypayoff}) has the unique Nash equilibrium $(x^*_{1}, x^*_{2})$ such that
\begin{equation}\label{uniqueequilibrium}
x^*_{1} = x^*_{2} = \frac{(a-c)\cosh{\gamma}}{1+ 2\mathrm{e}^{2\gamma}}. 
\end{equation}
\end{proposition}
\begin{proof}
The proof proceeds along the same lines as the proof in the classical Cournot duopoly \cite{peters}. We use the fact that a Nash equilibrium is a strategy profile $(x^*_{1}, x^*_{2})$ where $x^*_{1}$ and $x^*_{2}$ are mutually best replies. First, we determine the best reply function $\beta_{1}(x_{2})$ of player 1. It can be obtained by solving the maximization problem
\begin{equation}\label{argmax}
\argmax_{x_{1} \in [0,\infty)}u^{Q}_{1}(x_{1}, x_{2})~~\mbox{for a given value}~~x_{2} \geqslant 0.
\end{equation}
Let us first consider the case when $x_{2} \leqslant a\mathrm{e}^{-\gamma}$. Then the maximization problem~(\ref{argmax}) comes down to maximizing the function 
\begin{eqnarray}\label{czescfunkcji}
(x_{1}\cosh{\gamma} + x_{2}\sinh{\gamma})(a-c - \mathrm{e}^{\gamma}(x_{1} + x_{2})). 
\end{eqnarray}
For case $0 \leqslant x_{2} \leqslant (a-c)\mathrm{e}^{-2\gamma}\cosh{\gamma}$ the maximum point is $x_{1} = [(a-c)\cosh{\gamma} - \mathrm{e}^{2\gamma}x_{2}]/(\mathrm{e}^{2\gamma}+1)$. To show that $x_{1}$ is the unique best reply in this case let us note that $x_{2} \leqslant (a-c)\mathrm{e}^{-2\gamma}\cosh{\gamma} \leqslant (a-c)\cosh{\gamma}$ for $\gamma \geqslant 0$. Thus we have 
\begin{eqnarray}
\mathrm{e}^{\gamma}(x_{1} + x_{2}) = \frac{(a-c)\cosh{\gamma} + x_{2}}{2\cosh{\gamma}} \leqslant a-c.
\end{eqnarray}
This means that the player 1's payoff given by~(\ref{czescfunkcji}) is nonnegative. As a result, player~1 would make a loss by choosing $x_{1}$ such that $x_{1}+x_{2} > a\mathrm{e}^{-\gamma}$. For $(a-c)\mathrm{e}^{-2\gamma}\cosh{\gamma} < x_{2} \leqslant a\mathrm{e}^{-\gamma}$ function~(\ref{czescfunkcji}) of variable $x_{1}$ is strictly decreasing on $[0,\infty)$. Hence, it is optimal for player 1 to take $x_{1}=0$ and obtain $(a-c - \mathrm{e}^{\gamma}x_{2})x_{2}\sinh{\gamma}$. Note also that player 1 would not be willing to take $x_{1}$ such that $x_{1} + x_{2} > a\mathrm{e}^{-\gamma}$. Indeed, since $x_{2} \leqslant a\mathrm{e}^{-\gamma}$ and $\gamma \geqslant 0$, it is true that
\begin{align}
(a-c - \mathrm{e}^{\gamma}x_{2})x_{2}\sinh{\gamma} &= (a-\mathrm{e}^{\gamma}x_{2})x_{2}\sinh{\gamma} - cx_{2}\sinh{\gamma}\nonumber\\ &\geqslant -cx_{2}\sinh{\gamma} > -c(x_{1}\cosh{\gamma} + x_{2}\sinh{\gamma})
\end{align}
for each $x_{1} > 0.$

If $x_{2} > a\mathrm{e}^{-\gamma}$, problem~(\ref{argmax}) is equivalent to $\argmax_{x_{1} \in [0,\infty)}\{-c(x_{1}\cosh{\gamma} + x_{2}\sinh{\gamma})\}$. In this case it is optimal player 1 to choose $x_{1}=0$. Summarizing, player 1's best reply function is as follows
\begin{eqnarray}\label{bestreply1}
\beta_{1}(x_{2}) = \begin{cases}\frac{(a-c)\cosh{\gamma} - \mathrm{e}^{2\gamma}x_{2}}{\mathrm{e}^{2\gamma}+1} &\mbox{if}~ x_{2} \leqslant (a-c)\mathrm{e}^{-2\gamma}\cosh{\gamma} \\ 0 &\mbox{if}~ x_{2} > (a-c)\mathrm{e}^{-2\gamma}\cosh{\gamma}. \end{cases}
\end{eqnarray}
Similar arguments to those above show that the player 2's best reply function $\beta_{2}(x_{1})$ is given by formula
\begin{eqnarray}\label{bestreply2}
\beta_{2}(x_{1}) = \begin{cases}\frac{(a-c)\cosh{\gamma} - \mathrm{e}^{2\gamma}x_{1}}{\mathrm{e}^{2\gamma}+1} &\mbox{if}~ x_{1} \leqslant (a-c)\mathrm{e}^{-2\gamma}\cosh{\gamma} \\ 0 &\mbox{if}~ x_{1} > (a-c)\mathrm{e}^{-2\gamma}\cosh{\gamma}. \end{cases}
\end{eqnarray} 
We recall that a Nash equilibrium is a strategy profile at which each player chooses a best response to the other players' strategies. As a result, the Nash equilibria in the Li-Du-Massar scheme with (\ref{poprawionypayoff}) can be identified with the points of intersection of graphs determined by the best reply functions $\beta_{1}(x_{2})$ and $\beta_{2}(x_{1})$. The graphs are drawn in Fig~\ref{figure1}.
\begin{figure}[t]
\centering
\includegraphics[scale=1]{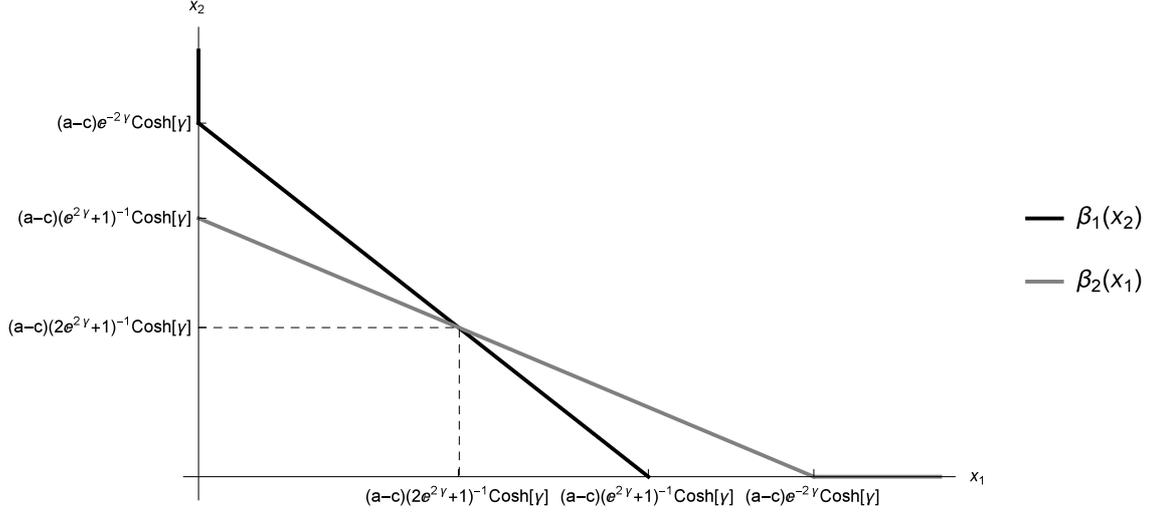}
\caption{The graphs of the best reply functions $\beta_{1}(x_{2})$ and $\beta_{2}(x_{1})$ given by~(\ref{bestreply1}) and~(\ref{bestreply2}). The intersection represents the unique Nash equilibrium in the game. \label{figure1}}
\end{figure}
In this way, there is the unique Nash equilibrium. It is obtained by solving the system of linear equations
\begin{eqnarray}
\begin{cases}x_{1} = \frac{(a-c)\cosh{\gamma} - \mathrm{e}^{2\gamma}x_{2}}{\mathrm{e}^{2\gamma}+1} \\ x_{2} = \frac{(a-c)\cosh{\gamma} - \mathrm{e}^{2\gamma}x_{1}}{\mathrm{e}^{2\gamma}+1}.\end{cases}
\end{eqnarray}
Thus, the solution coincides with the result provided in \cite{du}. Namely, $x^*_{1} = x^*_{2} = (a-c)(2\mathrm{e}^{2\gamma}+1)^{-1}\cosh{\gamma}.$ 
\end{proof}
It is worth noting that the quantum extension of the Cournot duopoly introduced in~\cite{du} affects the quantities $q_{1}$ and $q_{2}$ and makes use of the payoff function of the classically played duopoly. This implies that the set of payoff profiles generated by the Li-Du-Massar scheme with (\ref{poprawionypayoff}) is equal to one determined by function~(\ref{classicpayoff}). In particular, by solving the maximization problem
\begin{equation}
\argmax_{q_{1},q_{2} \in [0,\infty)}\left(u^{Q}_{1}(x_{1}, x_{2}) + u^{Q}_{2}(x_{1},x_{2})\right) = \argmax_{q_{1},q_{2} \in [0,\infty)}\mathrm{e}^{\gamma}(x_{1} + x_{2})\left(a-c - \mathrm{e}^{\gamma}(x_{1} + x_{2})\right)
\end{equation}
we obtain the set $\left\{(x_{1},x_{2})\colon x_{1} + x_{2} = (a-c)\mathrm{e}^{-\gamma}/2 \right\}$. Hence, for each $\gamma \in [0,\infty)$, a~symmetric Pareto optimal outcome $u^{Q}_{i}((a-c)\mathrm{e}^{-\gamma}/4, (a-c)\mathrm{e}^{-\gamma}/4)$ is equal to $(a-c)^2/8$ in both the classical and quantum case.
\section{Another example of the quantum Cournot duopoly scheme}
The Li-Du-Massar scheme with the refined payoff function~(\ref{poprawionypayoff}) is defined in accordance with the quantum protocols for finite games. The scheme generalizes the classicaly played Cournot duopoly and it keeps the set of feasible payoff profiles unchanged. We saw in subsection~\ref{sectionMW} that the scheme based on the Marinatto-Weber approach for bimatrix games does not satisfy the latter condition. The question now is whether these two requirements on a quantum scheme imply the unique quantum model for the Cournot duopoly problem. The two well-known and quite different quantum schemes for bimatrix games, introduced in \cite{ewl} and \cite{marinatto} suggest that the answer ought to be negative. In fact, this is the case. We can define another scheme that is consistent with the requirements above. In what follows, we give an example of a scheme that is similar in concept to the Li-Du-Massar scheme. The idea is based on `entangling' the players' quantities $x_{1}$ and $x_{2}$ in order to obtain $q_{1} = ax_{1} + bx_{2}$ and $q_{2} = ax_{2} + bx_{1}$ for some real numbers $a$ and $b$. 

Let $|\Psi_{\mathrm{in}}\rangle = |00\rangle$ be the initial state and $I(\gamma)$ be an entangling operator,
\begin{equation}\label{mojscheme1}
I(\gamma) =  \mathds{1}^{\otimes 2}\cos{\gamma} + \mathrm{i}\sigma^{\otimes 2}_{x}\sin{\gamma},~\gamma \in [0,\pi/4],
\end{equation} 
where $\mathds{1}$ and $\sigma_{x}$ are the identity and Pauli operator $X$, respectively, defined on $\mathbb{C}^2$. The resulting state is then given by
\begin{equation}
|\Psi_{\mathrm{fin}}\rangle = I(\gamma)|00\rangle = \cos{\gamma}|00\rangle + \mathrm{i}\sin{\gamma}|11\rangle. 
\end{equation}
Let us identify the player $i$'s strategies with $x_{i} \in [0,\infty)$ for $i=1,2$. The values $x_{1}$ and $x_{2}$ determine two positive operators $\left\{M_{1}(x_{1},x_{2}),M_{2}(x_{1},x_{2})\right\}$ given by formula
\begin{equation}
M_{i}(x_{1},x_{2}) = \begin{cases} x_{1}|0\rangle \langle 0| + x_{2}|1\rangle \langle 1| &\mbox{if}~i=1 \\ x_{2}|0\rangle \langle 0| + x_{1}|1\rangle \langle 1| &\mbox{if}~i=2.\end{cases}
\end{equation}
The measurement defined by $M_{i}(x_{1},x_{2})$ determines the quantities $q_{1}$ and $q_{2}$ in the following way:
\begin{equation}
q_{1} = \mathrm{tr}(M_{1}\rho_{1}),\quad q_{2} = \mathrm{tr}(M_{2}\rho_{2}),
\end{equation}
where $\rho_{1}$ and $\rho_{2}$ are the reduced density operators $\mathrm{tr}_{2}(|\Psi_{\mathrm{fin}}\rangle \langle \Psi_{\mathrm{fin}}|)$ and $\mathrm{tr}_{1}(|\Psi_{\mathrm{fin}}\rangle \langle \Psi_{\mathrm{fin}}|)$ of the first and second qubit, respectively. As a result
\begin{equation}
q_{1} = x_{1}\cos^2{\gamma} + x_{2}\sin^2{\gamma}, \quad q_{2} = x_{2}\cos^2{\gamma} + x_{1}\sin^2{\gamma}.
\end{equation}
Referring to function~(\ref{classicpayoff}) we obtain the following players' payoff functions:
\begin{equation}\label{mojscheme2}
u^{Q}_{1(2)}(x_{1}, x_{2}) = \begin{cases}\left(x_{1(2)}\cos^2{\gamma} + x_{2(1)}\sin^2{\gamma}\right)(a-c - x_{1} - x_{2}) &\mbox{if}~  x_{1} + x_{2} \leqslant a \\ -c\left(x_{1(2)}\cos^2{\gamma} + x_{2(1)}\sin^2{\gamma}\right) &\mbox{if}~ x_{1} + x_{2} > a.\end{cases}
\end{equation} 
Certainly, scheme~(\ref{mojscheme1})-(\ref{mojscheme2}) is a generalization of the classically played Cournot duopoly. If $\gamma = 0$ then payoff function (\ref{mojscheme2}) boils down to function (\ref{classicpayoff}). Since the scheme uses the classical payoff function, the corresponding set of feasible payoffs remains the same as in the classical case.
\begin{proposition}
If the marginal cost is positive, the set $N$ of Nash equilibria in the game defined by scheme~(\ref{mojscheme1})-(\ref{mojscheme2}) is as follows:
\begin{equation}
N = \begin{cases}\left\{\left(\frac{(a-c)\cos^2{\gamma}}{2\cos^2{\gamma} + 1}, \frac{(a-c)\cos^2{\gamma}}{2\cos^2{\gamma} + 1}\right)\right\} &\mbox{if}~\gamma \ne \pi/4 \\ \left\{\left(x,\frac{a-c}{2} - x\right), x \in \left[0, \frac{a-c}{2}\right]\right\} &\mbox{if}~\gamma = \pi/4. \end{cases}
\end{equation}
\end{proposition} 
\begin{proof}
The proof is similar in spirit to that of Proposition~\ref{prop1}. Let $x_{2}$ satisfy $x_{2} \leqslant a$. For $x_{2} \leqslant (a-c)\cos^2{\gamma}$ the solution of $\argmax_{q_{1} \in [0,\infty)}u^{Q}_{1}(x_{1},x_{2})$ is $x_{1} = [(a-c)\cos^2{\gamma} - x_{2}]/(2\cos^2{\gamma})$ as a global maximum of function of variable $x_{1}$,
\begin{equation}\label{1part}
\left(x_{1}\cos^2{\gamma} + x_{2}\sin^2{\gamma}\right)(a-c-x_{1} - x_{2}) 
\end{equation}
If $(a-c)\cos^2{\gamma} < x_{2} \leqslant a$ then function~(\ref{1part}) is monotonically increasing on $[0,\infty)$. Hence, in this case
\begin{equation}
\argmax_{\makebox[0pt]{$\scriptstyle x_{1} \in [0,\infty), x_{1} + x_{2} \leqslant a$}}u^{Q}_{1}(x_{1},x_{2}) = \{0\}.
\end{equation}
Since $x_{2} \leqslant a$, we have
\begin{equation}
u^{Q}_{1}(0,x_{2}) = x_{2}(a-x_{2})\sin^2{\gamma} - cx_{2}\sin^2{\gamma} \geqslant -cx_{2}\sin^2{\gamma} > -c\left(x_{1}\cos^2{\gamma} + x_{2}\sin^2{\gamma}\right)
\end{equation}
for each $x_{1} > 0$. It implies that
\begin{equation}
\argmax_{\makebox[0pt]{$\scriptstyle x_{1} \in [0,\infty), x_{1} + x_{2} \leqslant a$}}u^{Q}_{1}(x_{1},x_{2}) = \argmax_{\makebox[0pt]{$\scriptstyle x_{1} \in [0,\infty)$}}u^{Q}_{1}(x_{1},x_{2}).
\end{equation}
If $x_{2} > a$, then $u^{Q}_{1}(x_{1}, x_{2}) = -c\left(x_{1}\cos^2{\gamma} + x_{2}\sin^2{\gamma}\right)$. Hence, it is optimal for player 1 to choose $x_{1} = 0$. Summarizing, we have the following form of $\beta_{1}(x_{2})$: 
\begin{equation}\label{bestresponse3}
\beta_{1}(x_{2}) = \begin{cases}\frac{(a-c)\cos^2{\gamma} - x_{2}}{2\cos^2{\gamma}} &\mbox{if}~0\leqslant x_{2} \leqslant (a-c)\cos^2{\gamma} \\ 0 &\mbox{if}~x_{2} > (a-c)\cos^2{\gamma}.\end{cases}
\end{equation}
Similar arguments to those above show that
\begin{equation}\label{bestresponse4}
\beta_{2}(x_{1}) = \begin{cases}\frac{(a-c)\cos^2{\gamma} - x_{1}}{2\cos^2{\gamma}} &\mbox{if}~0\leqslant x_{1} \leqslant (a-c)\cos^2{\gamma} \\ 0 &\mbox{if}~x_{1} > (a-c)\cos^2{\gamma}.\end{cases}
\end{equation}
The best response function $\beta_{1}(x_{2})$ and $\beta_{2}(x_{1})$ are given in Fig~\ref{figure2}. 
\begin{figure}[t]
\centering
\includegraphics[scale=1]{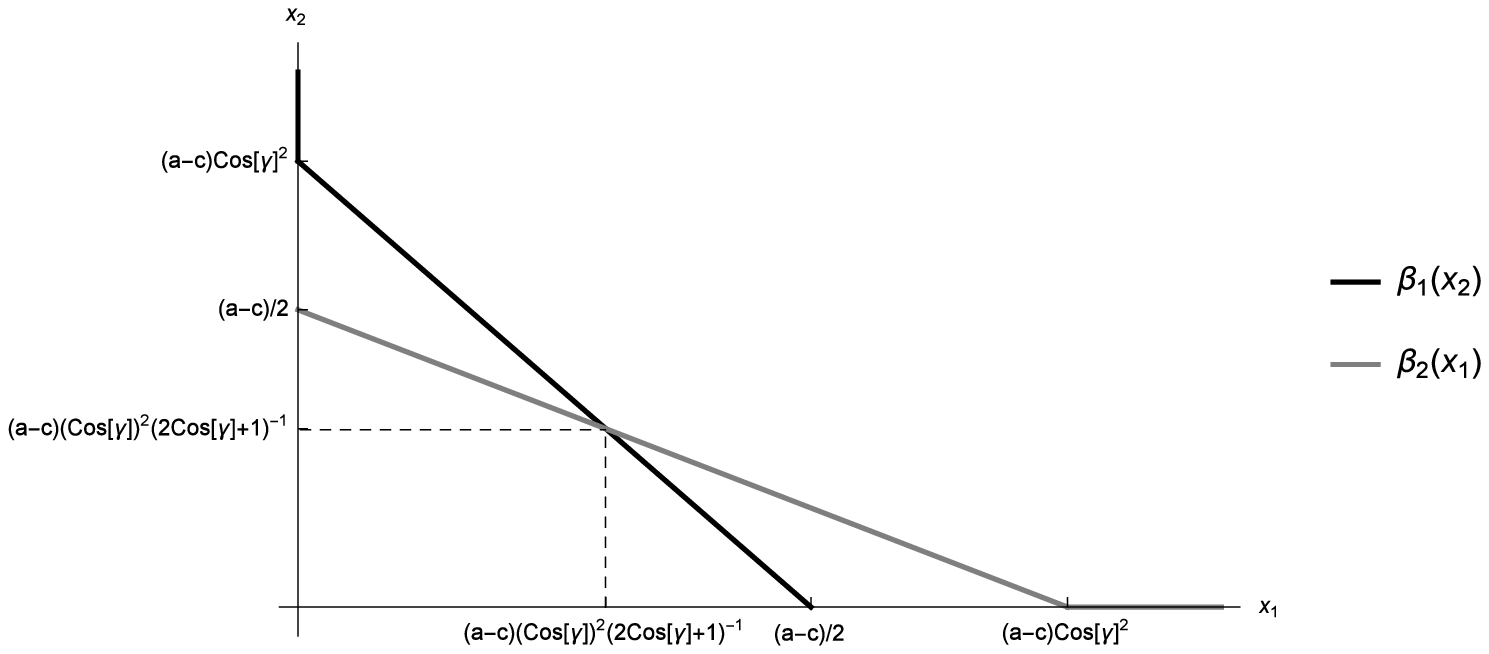}
\caption{The graphs of the best reply functions $\beta_{1}(x_{2})$ and $\beta_{2}(x_{1})$ given by~(\ref{bestresponse3}) and~(\ref{bestresponse4}). The intersection represents the unique Nash equilibrium in the game.\label{figure2}}
\end{figure}
Solving the system of equations determined by $\beta_{1}(x_{2})$ and $\beta_{2}(x_{1})$,
\begin{equation}
\begin{cases} x_{1} = \frac{(a-c)\cos^2{\gamma} - x_{2}}{2\cos^2{\gamma}} \\ x_{2} = \frac{(a-c)\cos^2{\gamma} - x_{1}}{2\cos^2{\gamma}}\end{cases}
\end{equation}
we obtain for $\gamma \ne \pi/4$ the unique solution $(x^*_{1}, x^*_{2})$, $x^*_{1} = x^*_{2} = [(a-c)\cos^2{\gamma}]/(2\cos^2{\gamma} + 1)$. If $\gamma = \pi/4$, there are infinitely many solutions $(x^*_{1},x^*_{2})$ such that $x^*_{1} + x^*_{2} = (a-c)/2$. This completes the proof.
\end{proof}
\begin{figure}[t]
\centering
\includegraphics[scale=0.9]{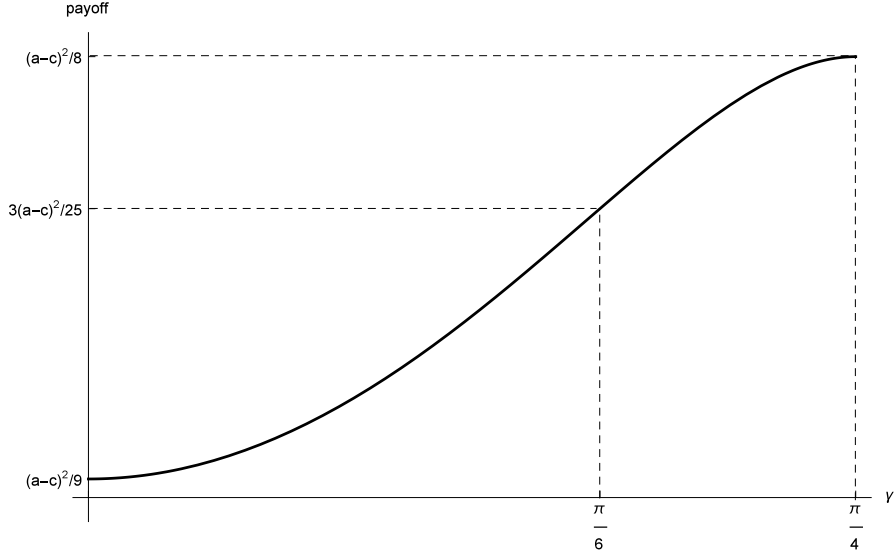}
\caption{The equilibrium payoff~(\ref{equilibriumpayoff}) depending on the entanglement parameter $\gamma$. \label{figure3}}
\end{figure}
The payoff corresponding to the Nash equilibrium depends on $\gamma$ and is given by
\begin{equation}\label{equilibriumpayoff}
u^{Q}_{i}(x^*_{1}, x^*_{2}) = \frac{(a-c)^2\cos^2{\gamma}}{2\cos^2{\gamma}+1},~\gamma \in \left[0,\frac{\pi}{4}\right].
\end{equation}
In particular, if the final state $|\Psi_{\mathrm{fin}}\rangle$ is maximally entangled ($\gamma = \pi/4$), the resulting equilibrium payoff $u^{Q}_{i}(x^*_{1}, x^*_{2})$ is Pareto optimal and equal to $(a-c)^2/8$.
\section{Conclusions}
The theory of quantum games has no rigorous mathematical structure. There is no formal axioms, definitions that would give clear directions of how a quantum game ought to look like. In fact, only one condition is taken into consideration. It says that a quantum game ought to include the classical way of playing the game. As a result, this allows one to define a quantum game scheme in many different ways. The schemes we have studied in section~\ref{section3} are definitely ingenious. They make a significant contribution to quantum game theory. Our work has shown which of the two schemes (the Iqbal-Toor scheme or the Li-Du-Masasr scheme) might be considered more reasonable. The payoffs in the game determined by the Iqbal-Toor scheme can go beyond the classical set of feasible payoffs compared with the Li-Du-Masasr scheme. Therefore, one may question whether the former scheme outputs the game played in a quantum manner or just another classical game. It might not mean that the latter scheme with (\ref{poprawionypayoff}) gives the definitive form of the quantum Cournot duopoly. At the end of the paper we defined another scheme in terms of quantum theory.

We can conclude that the question of quantum duopoly is still an open problem even in the simplest case of the Cournot duopoly. It is definitely worth investigating as other new schemes may bring us closer to specify a strict definition of a quantum game. 


\begin{thebibliography}{100}
\bibitem{mejer}  Meyer D A, Quantum Strategies, {\it Phys. Rev. Lett.} {\bf 82} 1052–55 (1999)
\bibitem{ewl} Eisert J Wilkens M and Lewenstein M, Quantum Games and Quantum Strategies
{\it Phys. Rev. Lett.} {\bf 83} 3077-80 (1999)
\bibitem{marinatto} Marinatto L and Weber T, A quantum approach to static games of complete information {\it Phys. Lett.} A {\bf 272} 291 (2000)
\bibitem{iqbalbackwards} Iqbal A and Toor A H, Backwards-induction outcome in a quantum game, Physical Review A {\bf 65} 052328 (2002)
\bibitem{du} Li H Du J and Massar S, Continuous-variable quantum games Phys. Lett. A {\bf 306} 73–8 (2002)
\bibitem{zhu} Zhu X and Kuang L M, The inﬂuence of entanglement and decoherence on the quantum Stackelberg duopoly game {\it J. Phys. A: Math. Theor.} {\bf 40} 7729-44 (2007)
\bibitem{zhu2} Zhu X and Kuang L M, Quantum Stackelberg duopoly game in depolarizing channel {\it Commun. Theor. Phys.} {\bf 49} 111-16 (2008)
\bibitem{bertrand} Khan S Ramzan M and Khan M K, Quantum model of Bertrand duopoly {\it Chinese Phys. Lett.} {\bf 27} 080302 (2010)
\bibitem{bertrand2} Khan S  Ramzan M and Khan M K, Quantum Stackelberg duopoly in the presence of correlated noise {\it J. Phys. A: Math. Theor.} {\bf 43} 375301 (2010)
\bibitem{doublekhan} Khan S and Khan K, Quantum Stackelberg Duopoly in a Noninertial Frame {\it Chinese Phys. Lett.} {\bf 28} 070202 (2011)
\bibitem{1.1} Du J Li H and Ju C, Quantum games of asymmetric information {\it Phys. Rev.} E {\bf 68} 016124 (2003)
\bibitem{2.1} Lo C F and Kiang D, Quantum Bertrand duopoly with differentiated products {\it Phys. Lett.} A {\bf 321} 94-8 2004
\bibitem{3.1} Qin G Chen X Sun M Zhou X and Du J, Appropriate quantization of asymmetric games with continuous strategies {\it Phys. Lett.} A {\bf 340} 78-86 (2005)
\bibitem{3.3} Du J Ju C and Li H, Quantum entanglement helps in improving economic efficiency {\it J. Phys. A: Math. Gen.} {\bf 38} 1559-65 (2005)
\bibitem{3.4} Chen X Qin G Zhou X and Du J, Quantum games of continuous distributed incomplete information {\it Chin. Phys. Lett.} {\bf 22} 1033 (2005)
\bibitem{4.1} Li Y Qin G Zhou X and Du J, The application of asymmetric entangled states in quantum games {\it Phys. Lett.} A {\bf 335} 447-51 (2006)
\bibitem{5} Wang X Yang X Miao L Zhou X and Hu C, Quantum Stackelberg duopoly of continuous distributed asymmetric information {\it Chin. Phys. Lett.} {\bf 24} 3040 (2007)
\bibitem{6} Sekiguchi Y Sakahara K and Sato T, Uniqueness of Nash equilibria in a quantum Cournot duopoly game {\it J. Phys. A: Math. Theor.} {\bf 43} 145303 (2010)
\bibitem{7} Li S, Simulation of continuous variable quantum games without entanglement {\it J. Phys. A: Math. Theor.} {\bf 44} 295302 (2011)
\bibitem{8.1} Yu R and Xiao R, Quantum Stackelberg Duopoly with isoelastic demand function {\it Journal of Computational Information Systems} {\bf 8} 3643-50 (2012)
\bibitem{8.2} Wang X and Hu C, Quantum Stackelberg duopoly with continuous distributed incomplete information {\it Chin. Phys. Lett.} {\bf 29} 120303 (2012)
\bibitem{8.3} Sekiquchi Y Sakahara K and Sato T, Existence of equilibria in quantum Bertrand-Edgeworth duopoly game {\it Quantum Inf. Process.} {\bf 11} 1371-9 (2012)
\bibitem{cournot} Cournot A, Researches into the mathematical principles of the theory of
wealth. Macmillan, New York (1897)
\bibitem{peters} Peters H 2008 Game Theory: A Multi-Leveled Approach, Springer-Verlag, Berlin
Heidelberg
\bibitem{hotel} Rahaman R, Majumdar P and Basu B, Quantum Cournot equilibrium for the Hotelling-Smithies model of product choice, J. Phys. A: Math. Theor. {\bf 45} 455301 (2012)
\bibitem{quantumbook} Merzbacher E, Quantum Mechanics 3rd Edition, John Wiley \& Sons, Inc. (1998)
\end{thebibliography}
\end{document}